\documentclass[11pt]{article}
\usepackage{amssymb} % set of additional math symbols
\usepackage{amsmath}     % new environments
\usepackage{amsthm}
\usepackage{todonotes}
\usepackage{mathtools}
\usepackage{a4wide}
\usepackage{graphicx}
\usepackage{hyperref}

\usepackage[ruled,vlined,linesnumbered]{algorithm2e}
\SetKw{KwNot}{not}
\SetKw{KwOR}{or}
\SetKw{KwAND}{and}
\SetKw{KwExitFor}{exit for-loop}
\SetKw{KwStep}{Step}

\usepackage{authblk}

\newcommand{\Rset}{\mathbb{R}}
\newcommand{\Xset}{\mathbb{X}}
\newcommand{\Yset}{\mathbb{Y}}
\newcommand{\Zset}{\mathbb{Z}}
\newcommand{\Qset}{\mathbb{Q}}
\newcommand{\Pset}{\mathbb{P}}

\newtheorem{thm}{Theorem}
\newtheorem{lem}{Lemma}

\newtheorem{cor}{Corollary}

\newtheorem{obs}{Observation}

\begin{document}

%*******TITLE AND AUTHORS*******************************************
\title{Computational Complexity of the Recoverable Robust Shortest Path Problem with Discrete Recourse}

\author[1]{Marcel Jackiewicz}
\author[1]{Adam Kasperski}
\author[1]{Pawe{\l} Zieli\'nski\footnote{Corresponding author}}

\affil[1]{
Wroc{\l}aw  University of Science and Technology, Wroc{\l}aw, Poland\\
            \texttt{\{marcel.jackiewicz,adam.kasperski,pawel.zielinski\}@pwr.edu.pl}}

\date{}

\maketitle

 \begin{abstract}
 In this paper, the recoverable robust shortest path problem is investigated. Discrete budgeted interval uncertainty representation is used to model uncertain second-stage arc costs. 
 The known complexity results for this problem are strengthened. Namely,
 it is shown that the recoverable robust shortest path problem is
  $\Sigma_3^p$-hard for the arc exclusion and arc symmetric difference neighborhoods.  Furthermore, it is also proven that the inner adversarial problem for these neighborhoods is $\Pi_2^p$-hard.
  \end{abstract}
  
\noindent \textbf{Keywords}: robust optimization, interval data, recovery, shortest path, computational complexity

\section{Preliminaries}

Let $G=(V,A)$ be a directed graph with nonnegative arc costs $C_e\geq 0$, $e\in A$. Two nodes $s, t\in V$ are distinguished as source and destination nodes, respectively. We will denote by~$\Phi$ the set of all simple $s$-$t$ paths in $G$. In the shortest path problem, we seek a simple $s$-$t$ path in $G$, minimizing the total cost.  This classic network optimization problem can be solved in polynomial time using one of several known algorithms~(see, e.g.,~\cite{AMO93}). In this paper, we 
wish to investigate
 a recoverable version of the problem, in which an $s$-$t$ path $X\in \Phi$, called a \emph{first-stage path}, can be modified to some extent in the second stage by applying a limited \emph{recovery action}.  This action consists in choosing a 
 \emph{second-stage path}~$Y$ in some prescribed neighborhood $\Phi(X,k)$ of $X$, where $k\geq 0$ is a given 
 \emph{recovery parameter}. We will consider the following three types of neighborhoods, 
known from the literature (see, e.g.,~\cite{SAO09,NO13}),
called \emph{arc inclusion}, \emph{arc exclusion}, and \emph{arc symmetric difference} neighborhoods, respectively: 
\begin{align}
\Phi^{\mathrm{incl}}(X,k)&=\{Y\in\Phi \,:\,  |Y\setminus X|\leq k\}, \label{incl}\\
\Phi^{\mathrm{excl}}(X,k)&=\{Y\in\Phi \,:\,  |X\setminus Y|\leq k\}, \label{excl}\\
\Phi^{\mathrm{sym}}(X,k)&=\{Y\in\Phi \,:\,  | (Y\setminus X) \cup (X\setminus Y) |\leq k\} \label{sym}.
\end{align}  
The cost of the first-stage path $X$ is known and equals $\sum_{e\in X} C_e$. The second-stage arc costs are uncertain while the path $X$ is being chosen. Let $\hat{c}_e$ be a \emph{nominal second-stage cost} of the arc $e\in A$ and $\Delta_e$ be a 
\emph{maximum deviation} of the second-stage cost of $e$ from $\hat{c}_e$. Therefore, $c_e\in [\hat{c}_e,\hat{c}_e+\Delta_e]$ for each $e\in A$. 
In the following, we will assume that each realization of the second-stage arc costs $(c_e)_{e\in A}$, called \emph{scenario}, belongs to the following \emph{uncertainty set}~\cite{BS04}:
\begin{align}
\mathcal{U}(\Gamma^d)&=\{\pmb{c}=(c_e)_{e\in A} \,:\, c_e\in [\hat{c}_e, \hat{c}_e+\Delta_e],
 |\{e\in A\,:\, c_e>\hat{c}_e\}|\leq \Gamma^d \}.
\label{intsetgd}
\end{align}
The set $\mathcal{U}(\Gamma^d)$ is called a \emph{discrete budgeted uncertainty set} and the parameter $\Gamma^d\in \{0,\dots,|A|\}$ is called a \emph{discrete budget}. It is the maximum number of the second-stage arc costs that can differ from their nominal values. 
In this paper, we study the following \emph{recoverable robust  shortest path problem} with discrete recourse~\cite{B12, NO13}:
 \begin{equation}
\textsc{Rec Rob SP} : \; \min_{X \in \Phi} \left ( \sum_{e\in X} C_e  + 
\max_{\pmb{c}\in \mathcal{U}(\Gamma^d)}\min_{Y\in \Phi(X,k)} \sum_{e\in Y} c_e \right).
\label{rrsp}
\end{equation}
Given a first-stage path $X\in \Phi$, we will consider the following \emph{adversarial problem}:
\begin{equation}
\textsc{Adv SP}: \; \max_{\pmb{c}\in \mathcal{U}(\Gamma^d)}\min_{Y\in \Phi(X,k)} \sum_{e\in Y} c_e.
\label{adv}
\end{equation}
Given additionally a second-stage cost scenario $(c_e)_{e\in A}$ we will also investigate the following \emph{incremental problem}:
\begin{equation}
\textsc{Inc SP}: \; \min_{Y\in \Phi(X,k)} \sum_{e\in Y} c_e.
\label{inc}
\end{equation}

The concept of recoverable robustness was first introduced in~\cite{LLMS09}. The recoverable robust version of the shortest path problem (\textsc{Rob Rec SP}) was first investigated in~\cite{B12}, where it was shown to be strongly NP-hard for the arc inclusion neighborhood.  It was shown in~\cite{SAO09} that \textsc{Inc SP} for the arc exclusion neighborhood is strongly NP-hard, and the same hardness result was established in~\cite{NO13} for the arc symmetric difference neighborhood. Interestingly, the problem \textsc{Inc SP} for the arc inclusion neighborhood can be solved in polynomial time~\cite{SAO09}.  The problem \textsc{Adv SP}  is strongly NP-hard for all 
the neighborhoods above, which results from a reduction from the \emph{most vital arcs problem}~\cite{BKS95} (see also~\cite{NO13}). 
Since \textsc{Adv SP} is a special case of \textsc{Rec Rob SP} (we set $C_e=0$ for $e\in X$ and $C_e=M$ for some large number $M$, if $e\notin X$), the same hardness result holds for \textsc{Rec Rob SP}. The \textsc{Rec Rob SP} problem is  known to be weakly NP-hard for the class of arc-series parallel digraphs~\cite{GLW22}. This result follows immediately from the fact that the recoverable robust representatives selection problem, discussed in~\cite{GLW22}, is equivalent to \textsc{Rec Rob SP} for a very restrictive arc series-parallel digraph.
 Finally, it has been recently shown in~\cite{KZ24d} that \textsc{Rec Rob SP} is strongly NP-hard and not approximable unless P=NP for acyclic layered digraphs. These negative results motivated us to place the problem in a higher complexity class.

The recoverable robust shortest path problem was also considered for the \emph{continuous budgeted uncertainty set}~\cite{NO13}, that is, the set resulting from replacing the condition 
$ |\{e\in A\,:\, c_e>\hat{c}_e\}|\leq \Gamma^d$ in~(\ref{intsetgd}) with its continuous version
 $\sum_{e\in A} (c_e-\hat{c}_e)\leq \Gamma^c$ (the parameter $\Gamma^c\in \Rset_{+}$ is called a \emph{continuous  budget}).
  In this case, the \textsc{Rec Rob SP} problem remains strongly NP-hard for acyclic digraphs~\cite{KZ24d}, but admits a compact (of polynomial size) mixed-integer programming
  formulation (MIP for short) for all three considered neighborhoods~\cite{JKZ25}. 
 Compact MIP formulations for \textsc{Rec Rob SP} also exist under~$\mathcal{U}(\Gamma^d)$ when $\Gamma^d=|A|$
 (see~\cite{JKZ25}). In this case, the set $\mathcal{U}(\Gamma^d)$ 
 reduces to the standard interval uncertainty representation $\mathcal{U}=\times_{e\in A} [\hat{c}_e, \hat{c}_e+\Delta_e]$. The problem then remains strongly NP-hard for general digraphs~\cite{B12, NO13} but can be solved in polynomial time when the input graph is acyclic~\cite{JKZ25}.

\paragraph{Our results} In this paper, we strengthen the negative complexity results for \textsc{Rec Rob SP} with discrete recourse, i.e. with the uncertainty set of the form~(\ref{intsetgd}). We show that the problem is  $\Sigma_3^p$-hard for the arc exclusion and the arc symmetric difference neighborhoods.  We also prove that the adversarial problem for these neighborhoods is $\Pi_2^p$-hard. Furthermore, for the arc exclusion neighborhood, the hardness results remain valid even if $k=2$ in case of adversarial and $k=3$ in case of recoverable robust problems.
 This demonstrates the great difficulty in solving the recoverable robust shortest path problem and the inner adversarial problem with a discrete recourse. In particular, they exclude compact (of polynomial size)
  MIP formulations for them (see, e.g.,~\cite{WO21}).
 Table~\ref{tabres} summarizes the results obtained in this paper. 
  \begin{table}[h]
 \begin{small}
 \caption{Summary of the results for  \textsc{Adv SP} and  \textsc{Rec Rob SP}
 under the discrete  budgeted uncertainty  $\mathcal{U}(\Gamma^d)$
 for various neighborhoods~$\Phi(X,k)$.}
 \label{tabres}
 \begin{center}

  \begin{tabular}{l|l|l|l}
                      \hline
                    &\multicolumn{3}{c}{Neighborhoods $\Phi(X,k)$}\\
                     \cline{2-4}
Problem& $\Phi^{\mathrm{incl}}(X,k)$ & $\Phi^{\mathrm{excl}}(X,k)$ & $\Phi^{\mathrm{sym}}(X,k)$\\ \hline \hline
    \textsc{Adv SP}      & strongly NP-hard~\cite{BKS95}& $\Pi_2^p$-hard, $\Pi_2^p$-hard &  $\Pi_2^p$-hard, $\Pi_2^p$-hard\\
         & &  to approximate,&  to approximate \\ 
         & &  even for $k=2$  &(Corollary~\ref{cdadsym})\\
         &  &   (Corollary~\ref{cdadexcl})                                        &\\
      \hline
 \textsc{Rec Rob SP}   &   strongly NP-hard,&  $\Sigma_3^p$-hard, $\Sigma^p_3$-hard & $\Sigma_3^p$-hard,  $\Sigma^p_3$-hard\\
                                     &   not approximable,  even for &   to approximate,  & to approximate\\
                                      &  $k=1$ and $\Gamma^d=1$~\cite{NO13} &  even for $k=3$ &(Corollary~\ref{cdrrsym})\\
                                      &                                      & (Corollary~\ref{cdrrexcl})&\\                                      
  \hline
 \end{tabular}
 \end{center}
\end{small}
 \end{table}

\section{Hardness of the adversarial problem}

To show the hardness results for \textsc{Adv SP} we will use the following problem:

\begin{description}
\item{$\forall (\Gamma) \exists$\textsc{CNF-SAT}}
\item{Input:} Two disjoint sets $\Yset$ and $\Zset$ of Boolean variables that take the values 0 and 1,
a Boolean formula $\digamma(\Yset,\Zset)$ in conjunctive normal form, an integer~$\Gamma$.
\item{Question:} Is it true that for all subsets $\Yset'\subseteq \Yset$ such that $|\Yset'|\leq \Gamma$,
there exists  0-1~assignment~$(\pmb{y},\pmb{z})$
of variables in~$\Yset$ and $\Zset$ in which  all variables in $\Yset'$ are set to~0
($\pmb{y}|_{\Yset'}=\pmb{0}$, where $\pmb{y}|_{\Yset'}$ stands for the restriction of the 
assignment $\pmb{y}$ to the set~$\Yset'$)
 and
$\digamma(\pmb{y},\pmb{z})$ is satisfied?
\end{description}

\begin{thm}
\label{thm01}
$\forall (\Gamma) \exists$\textsc{CNF-SAT} is $\Pi^p_2$-complete, even for $|\Yset|=|\Zset|$.
\end{thm}
The proof of Theorem~\ref{thm01} is shown in Appendix~\ref{dod}. It is adapted from~\cite[Theorem~1]{GLW24} to our case. We will also use the following problem:
\begin{description}
\item{\textsc{2-Vertex-Disjoint Paths}}
\item{Input:} A digraph $G=(V,A)$ and two terminal pairs $(s_i,t_i)$, where
$s_i,t_i\in V$, $i\in\{1,2\}$ are pairwise distinct vertices.
\item{Question:} Are there vertex-disjoint paths $\pi_1$ and $\pi_2$  in~$G$
such that $\pi_i$ is a simple $s_i$-$t_i$ path, $i\in \{1,2\}$?
\end{description}
The \textsc{2-Vertex-Disjoint Paths} problem is strongly NP-complete in general digraphs~\cite{FHW80}.

\begin{obs}
Let $\widehat{\mathcal{U}}(\Gamma^d)=\{\pmb{c}=(c_e)_{e\in A} \,:\, c_e\in \{\hat{c}_e,\overline{c}_e\},
\overline{c}_e=\hat{c}_e+\delta_e\Delta_e,
\delta_e\in\{0,1\},
 e\in A\,:\, \sum_{e\in A} \delta_e\leq \Gamma^d \}\subseteq \mathcal{U}(\Gamma^d)$. The following equation holds:
 \[
  \max_{\pmb{c}\in \mathcal{U}(\Gamma^d)}\min_{Y\in \Phi(X,k)} \sum_{e\in Y} c_e=
   \max_{\pmb{c}\in \widehat{\mathcal{U}}(\Gamma^d)}\min_{Y\in \Phi(X,k)} \sum_{e\in Y} c_e.
 \]
 \label{obext}
\end{obs}
In what follows, we can replace $\mathcal{U}(\Gamma^d)$ with $\widehat{\mathcal{U}}(\Gamma^d)$ in the definition of \textsc{Adv SP} and \textsc{Rec Rob SP}. From now on, we use the uncertainty set~$\widehat{\mathcal{U}}(\Gamma^d)$.
 Consider the following decision version of the adversarial problem:
\begin{description}
\item{\textsc{Decision-Adv SP}$(\Gamma^d)$}
\item{Input:}  A digraph~$G=(V,A)$,
a path $X\in \Phi$, second stage arc costs under the discrete budgeted
uncertainty~$\widehat{\mathcal{U}}(\Gamma^d)$, i.e. 
nominal costs~$\hat{\pmb{c}}=(\hat{c}_e)_{e\in A}\geq \pmb{0}$,  cost upper bounds~$\overline{\pmb{c}}=(\overline{c}_e)_{e\in A}\geq \hat{\pmb{c}}$, an integer budget~$\Gamma^d$,
a recovery parameter~$k$, a bound~$\gamma\geq 0$.
\item{Question:} Does for every $\pmb{c}\in \widehat{\mathcal{U}}(\Gamma^d)$ a path
$Y\in \Phi(X,k)$ exist such that $\sum_{e\in Y} c_e\leq \gamma$?
\end{description}

\begin{lem}
\textsc{Decision-Adv SP}$(\Gamma^d)$ 
with 
$ \Phi(X,k)\in \{\Phi^{\mathrm{excl}}(X,k),
\Phi^{\mathrm{sym}}(X,k)\}$
 is in the class~$\Pi^p_2$.
\label{lpi2p}
\end{lem} 
\begin{proof}
Indeed, for these two neighborhoods, \textsc{Decision-Adv SP}$(\Gamma^d)$ can be expressed as
 $(\forall \,\pmb{c}\in \widehat{\mathcal{U}}(\Gamma^d))(\exists\, Y\in \Phi(X,k))
(\sum_{e\in Y} c_e\leq \gamma)$.
\end{proof}

% (Tej uwagi nie ma na arxiv)
%\begin{rem}
%\textsc{Decision-Adv SP}$(\Gamma^d)$ 
%with 
%$ \Phi(X,k):=\Phi^{\mathrm{incl}}(X,k)$
% lies in coNP.  In this case, it 
% can be expressed as
%$(\forall \,\pmb{c}\in \widehat{\mathcal{U}}(\Gamma^d))
%( \min_{Y\in \Phi^{\mathrm{incl}}(X,k)} \sum_{e\in Y} c_e\leq \gamma)$,
%where $\min_{Y\in \Phi^{\mathrm{incl}}(X,k)} \sum_{e\in Y} c_e$ is 
% the incremental problem~(\ref{inc}), which can be solved for $\Phi^{\mathrm{incl}}(X,k)$  in 
%polynomial time~\cite{SAO09}.
% \end{rem}

\begin{thm}
\textsc{Decision-Adv SP}$(\Gamma^d)$ 
 with the neighborhood $\Phi(X,k):=\Phi^{\mathrm{excl}}(X,k)$
 is $\Pi^p_2$-complete even if $k=2$.
 \label{dadexcl}
\end{thm}
\begin{proof}
Lemma~\ref{lpi2p} shows that our problem lies in $\Pi^p_2$.
Now, we need to prove that it is $\Pi^p_2$-hard.
We will show a reduction  from $\forall (\Gamma) \exists$\textsc{CNF-SAT}
to \textsc{Decision-Adv SP}$(\Gamma^d)$ 
 with $\Phi^{\mathrm{excl}}(X,k)$.
 
We are given an instance $(\Yset,\Zset,\digamma,\Gamma)$ of $\forall (\Gamma) \exists$\textsc{CNF-SAT}, i.e.
 two  disjoint sets of variables, $\Yset=\{y_1,\ldots,y_n\}$ and $\Zset=\{z_1,\ldots,z_n\}$,
 a  formula  $\digamma$ in CNF,  $\digamma=\mathcal{C}_1\wedge\cdots\wedge\mathcal{C}_m$,
 a budget~$\Gamma$.
 The reduction is divided into two steps.  
 
 In the first step,
 we obtain a digraph~$G=(V, A)$ from $(\Yset,\Zset,\digamma, \Gamma)$  by a standard transformation used to reduce \textsc{CNF-SAT} to \textsc{2-Vertex-Disjoint Paths} (see~\cite{FHW80}).
 For each variable $y_i\in \Yset$, we associate a variable subgraph consisting of only 
 two internally node-disjoint paths (the internal nodes:
 $y_i^{0},\ldots, {y}_i^{m}$ and $\overline{y}_i^{0},\ldots,\overline{y}_i^{m}$)
 of arc length~$m$, called the upper and lower paths. 
 If we go through the upper (lower) path, we simulate assigning $y_i:= 0$ ($y_i:= 1$).
 Similarly, we associate $n$ variable subgraphs
 with the variables in~$\Zset$.
 For each clause $\mathcal{C}_j$, $j\in[m]$,
    we add a clause gadget consisting in 3 arcs with the same terminals~$\mathcal{C}^{\text{in}}_j$ and $\mathcal{C}^{\text{out}}_j$.
 Each of these arcs represents a literal~$l_r\in \mathcal{C}_j$, $l_r\in \{y_r,\overline{y}_r, z_r,\overline{z}_r\}$.
 The literal arcs are connected to appropriate arcs in the variable gadgets,
 i.e.,
 if $l_r=y_r$ then we connect it to the~$(y_r^{j-1}, y_r^{j})$ arc
 and if $l_r=\overline{y}_r$ then we connect it to the~$(\overline{y}_r^{j-1}, \overline{y}_r^{j})$ arc.
 We do the same if $l_r=z_r$ or $l_r=\overline{z}_r$, $r\in [n]$, just for~$\Zset$ variable gadgets.
 The connecting of literal arcs and variable gadgets is carried on using \emph{switch} gadgets
 which have a blocking property that allows the path traversing the gadgets to encode a Boolean assignment.
 
 We show the switch gadget in Figure~\ref{fig:switch-gadget}.
 \begin{figure}[htbp]
 \begin{center}
 \includegraphics[width=\textwidth]{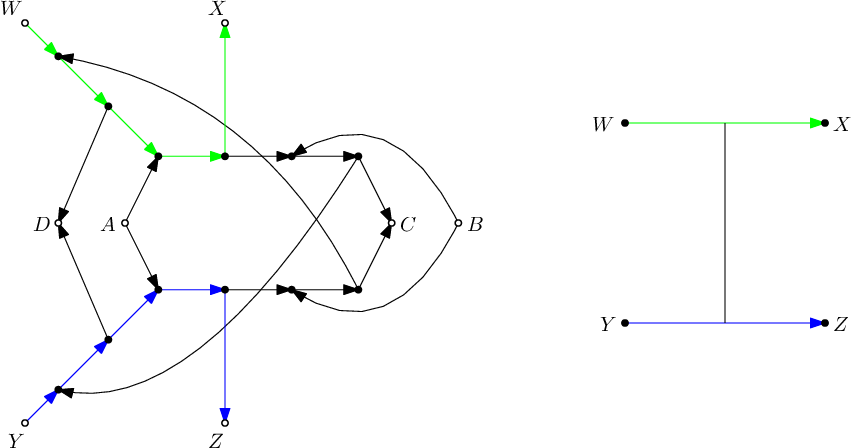}
 \caption{
     The switch gadget (left), with the property of blocking either the green path~$W-X$ or
     the blue path~$Y-Z$, and its simplified, symbolic depiction (right).
 }
 \label{fig:switch-gadget}
 \end{center}
 \end{figure}
 Each switch gadget has two entry nodes, denoted by $A$ and $B$, through which any path 
traversing the gadget may enter. It also contains two exit nodes, 
$C$ and $D$, and a path may use either of them to leave the gadget.
 The blocking property relies on the fact that if two vertex-disjoint paths traverse a switch gadget,
 one of them being an~$A-C$ path, the other a~$B-D$ path,
 then exactly one of the colored paths is unblocked
 (either the~$W-X$ or the~$Y-Z$ path), in the sense that some other vertex-disjoint path with respect to the paths can traverse it.
 The details of how this property is enforced can be found in \cite{FHW80}.
 We will use the schematic diagram of the gadget (right in Figure \ref{fig:switch-gadget}) to symbolize connecting two arcs using a switch gadget for readability of the final graph layout.
 The line connecting the colored arcs is not an arc but a symbolic depiction of the fact that the two arcs are connected, i.e.,
 they are expanded into the respective colored paths in a single switch gadgets.
 We now precisely describe the connecting of literal arcs and variable gadgets. 
 Consider a literal $l_r \in \mathcal{C}_j$ and the switch gadget 
associated with $l_r$. 
 If $l_r = y_r$, 
we replace the arc $(y_r^{j-1}, y_r^{j})$ with the path $W-X$ by 
identifying node $y_r^{j-1}$ with $W$ and node $y_r^{j}$ with $X$. 
Simultaneously, the literal arc corresponding to $l_r$ is replaced with 
the path $Y-Z$ by identifying node $\mathcal{C}^{\text{in}}_j$ with $Y$ 
and node $\mathcal{C}^{\text{out}}_j$ with $Z$. Similarly, if $l_r = \overline{y}_r$, we replace 
the arc $(\overline{y}_r^{j-1}, \overline{y}_r^{j})$ with the path $W-X$ 
and the corresponding literal arc with the path $Y-Z$ in an analogous 
manner.  We do the  same if $l_r=z_r$ or $l_r=\overline{z}_r$.

 Since there are is a total of~$3m$ literals in all clauses we need to use~$3m$ switches.
 The switches are linked in series by identifying the~$C$ and~$B$ nodes of a given gadget with the~$A$ and~$D$ nodes
 of the next gadget, respectively.
 We add the arc~$(D_{\text{first}}, y_1^{\text{in}})$ connecting the first switch with the first variable gadget
 and an arc~$(z_n^{\text{out}}, \mathcal{C}^{\text{in}}_1)$ connecting the last variable gadget with the first clause gadget.
 We set~$s_1 = B_{\text{last}}$, that is the~$B$ node of the last switch in series,
 and~$t_1 = \mathcal{C}^{\text{out}}_m$.
 Hence, an~$s_1$-$t_1$ path traverses all switches, all variable gadgets and all clause gadgets, in that order.
 We fix the second terminal pair $(s_2, t_2)$ as $s_2=A_{\text{first}}$ and $t_2=C_{\text{last}}$,
 which invokes the blocking property of the switches (the technical details are explained in \cite{FHW80}).
 The graph~$G$ is depicted in Figure~\ref{fadvgexl}.

 The validity of the reduction relies on the fact that~$\digamma$ is satisfiable if and only if $G$
 contains two node-disjoint simple paths, namely $\pi_1$ from $s_1$ to $t_1$ (the switch-variable-clause path)
 and $\pi_2$ from $s_2$ to $t_2$ (the switch path).
 Every satisfying 0-1 assignment $(\pmb{y},\pmb{z})$ of~$\digamma$
 corresponds to a simple paths $\pi_1$ which contains as a subpath the lower path of the~$y_r$ ($z_r$) variable,
 corresponding to the~$\overline{y}_r$ ($\overline{z}_r$) literal if and only if $y_r=1$ ($z_r=1$),  $r\in [n]$, in~$(\pmb{y},\pmb{z})$,
 since then any clause gadget corresponding to a clause satisfiable by setting $y_r=1$ ($z_r=1$) can be traversed using the upper path.
 To make sure the path traverses all variable and clause gadgets the switch gadgets were used.
 Therefore, the~$\pi_1$ path encodes the assignment and the switch path~$\pi_2$ enables the blocking property of the switches.

\begin{figure}[htbp]
\begin{center}
\includegraphics[width=\textwidth]{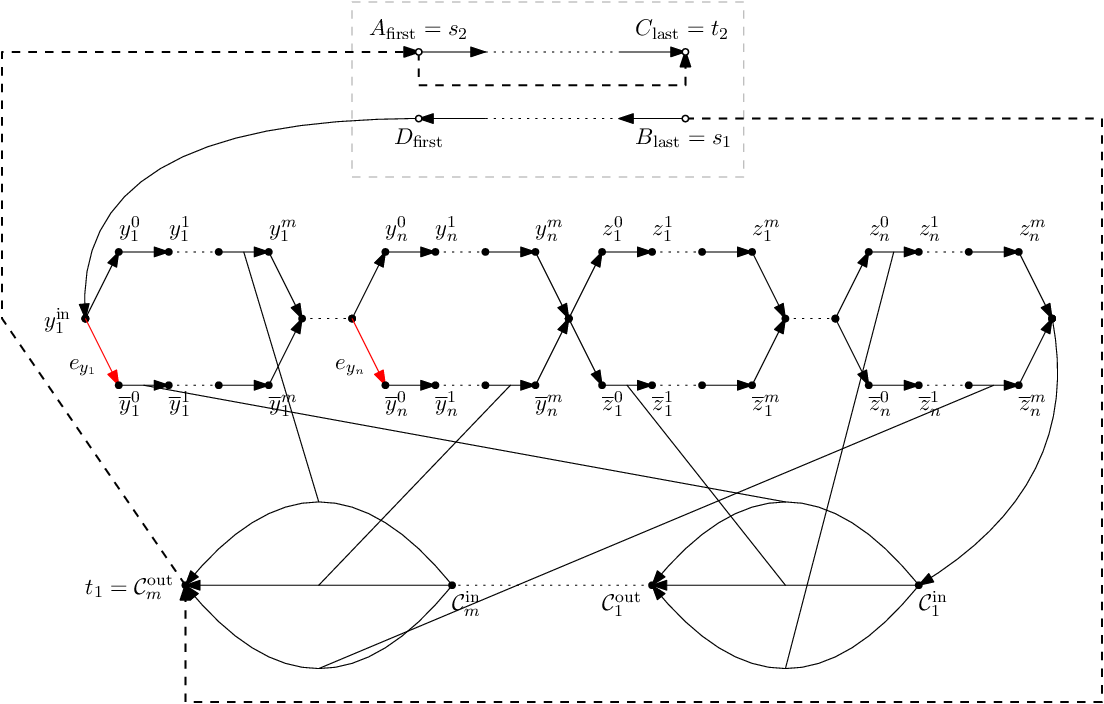}
\caption{A graph~$G$ in an instance of \textsc{Decision-Adv SP}$(\Gamma^d)$
         (with switch gadgets in the gray-dashed rectangle)
         corresponding to the formula $\digamma=(\overline{y}_1 \vee \overline{z}_1 \vee z_n)
         \wedge \cdots \wedge (y_1 \vee \overline{y}_n \vee \overline{z}_n)$.
         The graph layout resembles the presentation given in \cite{GW24}.
}
\label{fadvgexl}
\end{center}
\end{figure}

 In the second step we complete~$G=(V,A)$ by adding three arcs $(s_1,t_1)$,  $(t_1,s_2)$ and  $(s_2,t_2)$
 representing the path~$X\in\Phi$. We fix~$s:=s_1$ and $t:=t_2$ as the starting and destination nodes, respectively.
 In the arc set~$A$, we distinguish $n$ arcs labeled~$e_{y_i}$, $i\in[n]$,
 (see the arcs in red in Figure~\ref{fadvgexl}) and denote them by $A_{\pmb{y}}$.
The second stage arc costs under~$\widehat{\mathcal{U}}(\Gamma^d)$ are defined as follows:
$\hat{c}_e=0$ if $e\in A\setminus \{(s_1,t_1),(s_2,t_2)\}$,
$\hat{c}_e=1$ otherwise;
$\overline{c}_e=0$ if $e\in A\setminus \{(s_1,t_1),(s_2,t_2)\}\setminus A_{\pmb{y}}$,
$\overline{c}_e=1$ otherwise. Thus,
the uncertainty is only allowed on the arcs in $A_{\pmb{y}}$.
Finally, we set:
  the budget~$\Gamma^d:=\Gamma$, the  recovery parameter~$k:=2$ and
the bound~$\gamma:=0$. This completes the reduction.

What is left to show is that an
instance $(\Yset,\Zset,\digamma,\Gamma)$  of $\forall (\Gamma) \exists$\textsc{CNF-SAT}
 is a Yes-instance  if and only if  the built instance $(G,X,\widehat{\mathcal{U}}(\Gamma^d),k=2,\gamma=0)$
 of \textsc{Decision-Adv SP}$(\Gamma^d)$ with $\Phi^{\mathrm{excl}}(X,k)$ is a Yes-instance.

 ($\Rightarrow$)
 Assume that an instance $(\Yset,\Zset,\digamma,\Gamma)$ is a Yes-instance.
 Consider any scenario $\pmb{c}\in \widehat{\mathcal{U}}(\Gamma^d)$. Let $A'_{\pmb{y}}\subseteq A_{\pmb{y}}$ be the subset of arcs in $A_{\pmb{y}}$ whose costs under $\pmb{c}$ are at their upper bounds equal to~1.  Of course, $|A'_{\pmb{y}}|\leq \Gamma^d$.
 Let $\Yset'=\{y_i\, :\, e_{y_i}\in A'_{\pmb{y}}\}\subseteq \Yset$. Thus,
 $|\Yset'|=|A'_{\pmb{y}}|\leq  \Gamma^d=\Gamma$.
 Since $(\Yset,\Zset,\digamma,\Gamma)$ is a Yes-instance,
 there exists a 0-1 assignment $(\pmb{y},\pmb{z})$ in $\Yset\cup \Zset$,
 in which  $\pmb{y}|_{\Yset'}=\pmb{0}$, satisfying~$\digamma(\pmb{y},\pmb{z})$.
 Hence $(\pmb{y},\pmb{z})$ corresponds to a simple path~$\pi_1$ (the variable-clause path)
 together with a path~$\pi_2$ (the switch path), together node-disjoint,
 from~$s_1$ to~$t_1$ and $s_2$ to~$t_2$, respectively. 
 The paths $\pi_1$ and $\pi_2$ together with the arc~$(t_1,s_2)$ form a simple path~$Y\in \Phi^{\mathrm{excl}}(X,2)$.
 Observe that the path~$\pi_1$ does not traverse the~$A'_{\pmb{y}}$ arcs, as it uses only arcs that correspond to true literals
 while traversing the variable gadgets so that the negative literal arcs are available when traversing the clause gadgets
 (corresponding to choosing the negative literals for~$\Yset'$ variables).
 Thus, $\sum_{e \in Y} c_e=0$.
 This shows that $(G,X,\widehat{\mathcal{U}}(\Gamma^d),k=2,\gamma=0)$  is a Yes-instance.

 ($\Leftarrow$)
 Assume that an instance $(G,X,\widehat{\mathcal{U}}(\Gamma^d),k=2,\gamma=0)$ is a Yes-instance.
 Consider any $\Yset'\subseteq \Yset$ such that $|\Yset'|\leq \Gamma$.
 Let $A'_{\pmb{y}}=\{e_{y_i}\, :\, y_i\in \Yset'\}\subseteq A_{\pmb{y}}$.
 Obviously, $|A_{\pmb{y}}|\leq  \Gamma= \Gamma^d$.
 Consider scenario $\pmb{c}\in \widehat{\mathcal{U}}(\Gamma^d)$ defined as follows:
 $c_e:=\overline{c}_e$ if $e\in A'_{\pmb{y}}$, $c_e:=\hat{c}_e$ if $e\in A\setminus  A'_{\pmb{y}}$.
 Since $(G,X,\widehat{\mathcal{U}}(\Gamma^d),k=2,\gamma=0)$ is Yes-instance,
 there exists a simple path $Y\in \Phi(X,2)$ such that $\sum_{e\in Y} c_e=0$.
 Hence, no arc from $A'_{\pmb{y}} \cup \{(s_1,t_1),(s_2,t_2)\}$ belongs to~$Y$.
 It must traverse $(t_1,s_2)$ and contain as subpaths
 two simple node-disjoint paths~$\pi_1$ (the variable-clause path)  and~$\pi_2$
 (the switch path)  
 from~$s_1$ to~$t_1$ and $s_2$ to~$t_2$, respectively. 
 The variable-clause path~$\pi_1$ defines  the 0-1 assignment $(\pmb{y},\pmb{z})$, namely,
 $y_i:=1$ ($z_i:=1$)  if~$\pi_1$ traverses the negative literal subpath of the $y_i$ ($z_i$) variable gadget; 
 $y_i:=0$ ($z_i:=0$)  if~$\pi_1$ traverses the positive literal subpath of the $y_i$ ($z_i$) variable gadget.
 Clearly, $\pmb{y}|_{\Yset'}=\pmb{0}$.
 The paths $\pi_1$ and $\pi_2$ are node-disjoint. Thus,
 the blocking property of the switch gadgets guarantees that the variable-clause path~$\pi_1$
 traverses all variable and clause gadgets, i.e., all clauses are satisfied.
 Therefore, the 0-1 assignment $(\pmb{y},\pmb{z})$ satisfies~$\digamma(\pmb{y},\pmb{z})$
 and $(\Yset,\Zset,\digamma,\Gamma)$ is a Yes-instance.
\end{proof}

By Observation~\ref{obext} and Theorem~\ref{dadexcl}, we get the following corollary:
\begin{cor}
\label{cdadexcl}
\textsc{Adv SP} under $\mathcal{U}(\Gamma^d)$
 with the neighborhood $\Phi^{\mathrm{excl}}(X,k)$
 is $\Pi^p_2$-hard and $\Pi^p_2$-hard to approximate, even if $k=2$.
\end{cor}

\begin{thm}
\textsc{Decision-Adv SP}$(\Gamma^d)$ 
 with the neighborhood $\Phi(X,k):=\Phi^{\mathrm{sym}}(X,k)$
 is $\Pi^p_2$-complete.
 \label{dadsym}
\end{thm}
\begin{proof}
Lemma~\ref{lpi2p} shows that the problem under consideration lies in $\Pi^p_2$.
The proof is completed by showing that
it is $\Pi^p_2$-hard.
In the following, we provide a reduction  from $\forall (\Gamma) \exists$\textsc{CNF-SAT}
to \textsc{Decision-Adv SP}$(\Gamma^d)$ 
 with $\Phi^{\mathrm{sym}}(X,k)$.
 
Given  an instance $(\Yset,\Zset,\digamma,\Gamma)$ of $\forall (\Gamma) \exists$\textsc{CNF-SAT}
we construct the instance of  \textsc{Decision-Adv SP}$(\Gamma^d)$ 
 in much the same way as in the proof of Theorem~\ref{dadexcl} --- the only difference is in the second step of the construction. In the second step, 
we add to~$G=(V,A)$ two arcs $(s_1,t_1)$, $(s_2,t_2)$,
  and instead of  $(t_1,s_2)$ (see Figure~\ref{fadvgexl}, the dashed arc)
  we add a path connecting~$t_1$ with~$s_2$ consisting of $|A|$ arcs.
  There are~$3m$ switches and each switch contains 24 arcs (see Figure \ref{fig:switch-gadget});
  $2n$~variable gadgets with $2m + 4$ arcs each and two additional arcs: one connecting the series of switches and variables,
  the other connecting the variables and clauses.
  All literal arcs in clause gadgets are expanded into paths inside switches, and the same happens for~$3m$ out of~$4nm$ literal arcs in variable gadgets.
  Therefore, we have $|A| = 2n(2m + 4) + 72m + 2 - 3m = 4nm + 8n + 69m + 2$.
  We fix~$s:=s_1$ and $t:=t_2$ as the starting and the destination nodes, respectively.
  The added arcs represent the path~$X\in\Phi$.
  The set $A_{\pmb{y}}$ contains $n$ arcs labeled by~$y_i$, $i\in[n]$
 (see Figure~\ref{fadvgexl}, the red arcs). 
The second stage arc costs under~$\widehat{\mathcal{U}}(\Gamma^d)$ are defined as follows:
$\hat{c}_e=0$ if $e\in A\setminus \{(s_1,t_1),(s_2,t_2)\}$,
$\hat{c}_e=1$ otherwise;
$\overline{c}_e=0$ if $e\in A\setminus \{(s_1,t_1),(s_2,t_2)\}\setminus A_{\pmb{y}}$,
$\overline{c}_e=1$ otherwise. Finally, we set:  
  the budget~$\Gamma^d:=\Gamma$, the  recovery parameter~$k:=|A| + 2$, and
the bound~$\gamma:=0$. This completes the instance~$(G,X,\widehat{\mathcal{U}}(\Gamma^d),k,\gamma)$.

As in the proof of Theorem~\ref{dadexcl}, we can prove the following:
 an
instance $(\Yset,\Zset,\digamma,\Gamma)$  of $\forall (\Gamma) \exists$\textsc{CNF-SAT}
 is a Yes-instance  if and only if  the built instance $(G,X,\widehat{\mathcal{U}}(\Gamma^d),k,\gamma)$
 of \textsc{Decision-Adv SP}$(\Gamma^d)$ with $\Phi^{\mathrm{sym}}(X,k)$
  is a Yes-instance.
\end{proof}
Again, we have the following corollary, by Observation~\ref{obext} and Theorem~\ref{dadsym}.
\begin{cor}
\label{cdadsym}
\textsc{Adv SP} under $\mathcal{U}(\Gamma^d)$
 with the neighborhood $\Phi^{\mathrm{sym}}(X,k)$
 is $\Pi^p_2$-hard  and  $\Pi^p_2$-hard to approximate.
\end{cor}

\section{Hardness of the recoverable robust problem}

To characterize the complexity of \textsc{Rec Rob SP}, we will use the following problem:

\begin{description}
\item{$\exists \forall (\Gamma) \exists$3\textsc{CNF-SAT}}
\item{Input:} Three disjoint sets $\Xset$, $\Yset$, and $\Zset$ of Boolean variables that take the values 0 and 1,
a Boolean formula $\digamma(\Xset,\Yset,\Zset)$ in the conjunctive normal form, where every clause consists of
three literals,
 an integer~$\Gamma$.
\item{Question:} Is it true that
there is an 0-1~assignment $\pmb{x}$ of the variables in~$\Xset$ such that
for all subsets $\Yset'\subseteq \Yset$ such that $|\Yset'|\leq \Gamma$,
there exists  0-1~assignment~$(\pmb{y},\pmb{z})$
of variables in~$Y$ and $Z$ in which  all variables in $\Yset'$ are set to~0
($\pmb{y}|_{\Yset'}=\pmb{0}$)
 and
$\digamma(\pmb{x},\pmb{y},\pmb{z})$ is satisfied?
\end{description}
The problem
$\exists \forall (\Gamma) \exists$3\textsc{CNF-SAT} is $\Sigma^p_3$-complete, even if 
$|\Xset|=|\Yset|=|\Zset|$~\cite[Theorem~1]{GLW24}.
Consider the following decision version of \textsc{Rec Rob SP}:
\begin{description}
\item{\textsc{Decision-Rec Rob SP}$(\Gamma^d)$}
\item{Input:}  A digraph~$G=(V,A)$, first stage arc costs $\pmb{C}=(C_e)_{e\in A}\geq \pmb{0}$,
 second stage arc costs under the discrete budgeted
uncertainty~$\widehat{\mathcal{U}}(\Gamma^d)$, i.e.,
nominal costs~$\hat{\pmb{c}}=(\hat{c}_e)_{e\in A}\geq \pmb{0}$, upper bounds on the arc costs~$\overline{\pmb{c}}=(\overline{c}_e)_{e\in A}\geq \hat{\pmb{c}}$, an integer budget~$\Gamma^d$,
recovery parameter~$k$, a bound~$\gamma\geq 0$.
\item{Question:} 
 Is there a path $X\in \Phi$ such that for every $\pmb{c}\in \widehat{\mathcal{U}}(\Gamma^d)$  there exists
a  path $Y\in \Phi(X,k)$ such that $\sum_{e\in X}C_e+\sum_{e\in Y} c_e\leq \gamma$?
 \end{description}

\begin{lem}
\textsc{Decision-Rec Rob SP}$(\Gamma^d)$ with 
$ \Phi(X,k)\in \{\Phi^{\mathrm{excl}}(X,k),
\Phi^{\mathrm{sym}}(X,k)\}$ is in the class $\Sigma^p_3$.
\label{lrrpi2p}
\end{lem} 
\begin{proof}
Indeed, \textsc{Decision-Rec Rob SP}$(\Gamma^d)$ can be expressed in the form:
 $(\exists\, X\in \Phi)(\forall \,\pmb{c}\in \widehat{\mathcal{U}}(\Gamma^d))(\exists\, Y\in \Phi(X,k))
(\sum_{e\in X}C_e+\sum_{e\in Y} c_e\leq \gamma)$.
\end{proof}

% Nie ma tego na arxiv
%\begin{rem}
%\textsc{Decision-Rec Rob SP}$(\Gamma^d)$ with
%$ \Phi(X,k):=\Phi^{\mathrm{incl}}(X,k)$
% lies in the class $\Sigma^p_2$. In this case, it 
% can be expressed as
%$(\exists\, X\in \Phi)(\forall \,\pmb{c}\in \widehat{\mathcal{U}}(\Gamma^d))
%(\sum_{e\in X}C_e+\min_{Y\in \Phi^{\mathrm{incl}}(X,k)} \sum_{e\in Y} c_e\leq \gamma)$,
%where $\min_{Y\in \Phi^{\mathrm{incl}}(X,k)} \sum_{e\in Y} c_e$ is the incremental problem~(\ref{inc})
%that can be solved 
% in 
%polynomial time~\cite{SAO09}.
% \end{rem}

\begin{thm}
\textsc{Decision-Rec Rob SP}$(\Gamma^d)$ 
 with the neighborhood $\Phi(X,k):=\Phi^{\mathrm{excl}}(X,k)$
 is $\Sigma^p_3$-complete, even if $k=3$.
 \label{drrexcl}
\end{thm}
\begin{proof}
From Lemma~\ref{lrrpi2p} it follows that \textsc{Decision-Rec Rob SP}$(\Gamma^d)$
with $\Phi^{\mathrm{excl}}(X,k)$ is in~$\Sigma^p_3$.
We now prove that it is $\Sigma^p_3$-hard.

Let $(\Xset,\Yset,\Zset,\digamma,\Gamma)$ be an instance of $\exists \forall (\Gamma) \exists$3\textsc{CNF-SAT}.
In much the same way as in the proof of Theorem~\ref{dadexcl}, we construct 
an instance $(G,\pmb{C},\widehat{\mathcal{U}}(\Gamma^d),k,\gamma)$
of \textsc{Decision-Rec Rob SP}$(\Gamma^d)$. Small changes are in the first and second steps.

In the first step, we additionally associate $n$ variable subgraphs (see Figure~\ref{frobgexl}) with the variables in~$\Xset$, 
which results in adding them to the graph~$G$ depicted in Figure~\ref{fadvgexl}.
Then, for each clause, we additionally add to~$G$ the literal arcs corresponding to~$x_i$ or $\overline{x}_i$, $i\in[n]$ (in the same way as for the variables in $\mathbb{Y}$ and $\mathbb{Z}$).
Now the variable-clause path~$\pi_1$ must also traverse the~$\mathbb{X}$ variables gadgets
and can traverse clause gadgets using literal arcs associated with these variables.
\begin{figure}[htbp]
\begin{center}
\includegraphics{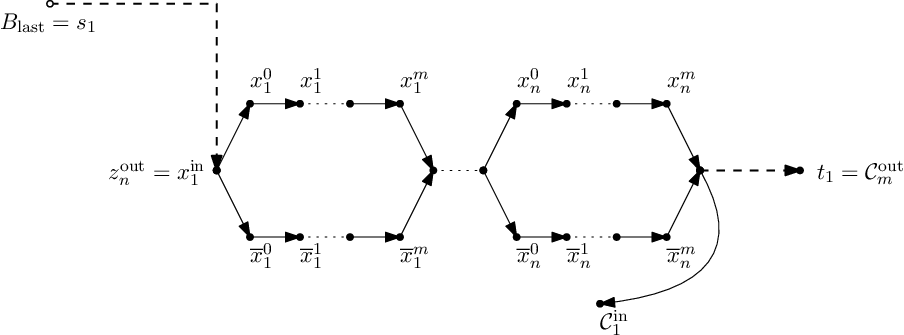}
    \caption{The placement of the $n$ $\pmb{x}$-variable subgraphs with respect to other components in Figure \ref{fadvgexl}.}
\label{frobgexl}
\end{center}
\end{figure}

In the second step, instead of three
    there are now four dashed arcs: $(s_1, x_1^{\text{in}})$ and $(x_n^{\text{out}}, t_1)$ and~$(t_1, s_2)$, $(s_2, t_2)$ as they were.
This finishes building~$G=(V,A)$.
We fix~$s:=s_1$ and $t:=t_2$ as the starting and the destination nodes, respectively.
The first stage costs are defined as follows:
$C_e=0$ if  $e\in  \{(s_1,x_1^{\text{in}}),(x_n^{\text{out}}, t_1),(t_1,s_2),(s_2,t_2)\}$
or $e$ is an arc of the $\pmb{x}$-variable subgraphs 
(see Figure~\ref{frobgexl}). The rest of the arcs have the first stage costs equal to~1.
The second stage arc costs under~$\widehat{\mathcal{U}}(\Gamma^d)$ are defined as follows:
$\hat{c}_e=0$ if $e\in A\setminus \{(s_1, x_1^{\text{in}}), (x_n^{\text{out}}, t_1),(s_2,t_2)\}$,
$\hat{c}_e=1$ otherwise;
$\overline{c}_e=0$ if $e\in A\setminus \{(s_1, x_1^{\text{in}}), (x_n^{\text{out}}, t_1),(s_2,t_2)\}\setminus A_{\pmb{y}}$,
$\overline{c}_e=1$ otherwise. Again, the uncertainty affects only the arcs from $A_{\pmb{y}}$.
Finally, we set:  
  the budget~$\Gamma^d:=\Gamma$, the  recovery parameter~$k:=3$, and
the bound~$\gamma:=0$. This completes the instance~$(G,\pmb{C},\widehat{\mathcal{U}}(\Gamma^d),k,\gamma)$.

We now need to prove the following assertion:
an instance $(\Xset,\Yset,\Zset,\digamma,\Gamma)$ of $\exists \forall (\Gamma) \exists$3\textsc{CNF-SAT}
 is a Yes-instance  if and only if  the corresponding instance $(G,\pmb{C},\widehat{\mathcal{U}}(\Gamma^d),k,\gamma)$
 of  \textsc{Decision-Rec Rob SP}$(\Gamma^d)$ with $\Phi^{\mathrm{excl}}(X,k)$
  is a Yes-instance.

($\Rightarrow$)
 Assume that  $(\Xset,\Yset,\Zset,\digamma,\Gamma)$ is a Yes-instance.
 Fix a 0-1 assignment~$\pmb{x}$ of the variables in~$\Xset$ such that for every 
 $\Yset'\subseteq \Yset$, $|\Yset'|\leq \Gamma$,
 there exists  a 0-1~assignment~$(\pmb{y},\pmb{z})$ of the variables in~$\Yset$ and $\Zset$, in which
 $\pmb{y}|_{\Yset'}=\pmb{0}$ and $\digamma(\pmb{x},\pmb{y},\pmb{z})$ is satisfied.
 A path~$X\in \Phi$ consists of the arcs: $(s_1, x_1^{\text{in}})$, $(x_n^{\text{out}}, t_1)$, $(t_1,s_2)$ and  $(s_2,t_2)$
    and the subpath traversing the~$\Xset$ variables gadgets according to the assignment~$\pmb{x}$, i.e.,
    if~$x_i := 0$, then~$X$ traverses the positive literal subpath (allowing the clauses to be traversed using negative literal arcs);
    otherwise, $X$~traverses the negative literal subpath.
 Clearly, first stage cost of~$X$ equals~0.
 Consider any scenario $\pmb{c}\in \widehat{\mathcal{U}}(\Gamma^d)$. Under this scenario a subset  of arcs $A'_{\pmb{y}}\subseteq A_{\pmb{y}}$, $|A'_{\pmb{y}}|\leq \Gamma^d=\Gamma$, have costs equal to~1 and the costs of the remaining arcs are~0.
 Define $\Yset'=\{y_i\, :\, e_{y_i}\in A'_{\pmb{y}}\}\subseteq \Yset$,
 $|\Yset'|\leq  \Gamma^d=\Gamma$. Since
$(\Xset,\Yset,\Zset,\digamma,\Gamma)$ is a Yes-instance,
 there is a 0-1 assignment $(\pmb{y},\pmb{z})$ in $\Yset\cup \Zset$,
 in which  $\pmb{y}|_{\Yset'}=\pmb{0}$ satisfying, together with~$\pmb{x}$, $\digamma(\pmb{x},\pmb{y},\pmb{z})$. The assignment
 $(\pmb{x},\pmb{y},\pmb{z})$ uniquely determines the subpath of $Y\in \Phi^{\mathrm{excl}}(X,3)$
 used to traverse the variables and clauses gadgets,
 i.e. the simple path from~$s_1$ to~$t_1$ denoted by~$\pi_1$.
 Indeed, $\pi_1$ contains
 the $3n$ subpaths (the lower or upper paths)
 traversing the arcs 
 corresponding to the opposite literals under~$(\pmb{x},\pmb{y},\pmb{z})$ in the $3n$ variable subgraphs.
 Then, the path uses the~$(z_n^{\text{out}}, C_1^{\text{in}})$ arc to enter the series of clause gadgets
 and traverses them by using the literal arcs according to the~$(\pmb{x},\pmb{y},\pmb{z})$ assignment.
 Thus $\pi_1$ does not contain any arc from~$A'_{\pmb{y}}$ and in consequence $\sum_{e \in \pi_1}c_e=0$.
 Since~$\pi_1$ is simple and traverses all variable and clause gadgets,
 it can be completed to traverse all switch gadgets at the beginning together with the other~$\pi_2$ path going from~$s_2$ to~$t_2$.
 Together they create a node-disjoint pair.
 We connect~$\pi_1$ and~$\pi_2$ to obtain the path~$Y \in \Phi^{\mathrm{excl}}(X,3)$.
 Indeed, $Y$ excludes only three arcs of~$X$: $\{(s_1,x_1^{\text{in}}),(x_n^{\text{out}}, t_1), (s_2,t_2)\}$.
 Moreover, $\sum_{e\in X} C_e+\sum_{e\in Y} c_e=0$.
  Hence $(G,\pmb{C},\widehat{\mathcal{U}}(\Gamma^d),k=3,\gamma=0)$  is a Yes-instance.

 ($\Leftarrow$) Assume that $(G,\pmb{C},\widehat{\mathcal{U}}(\Gamma^d),k=3,\gamma=0)$  is a Yes-instance.
 Thus, there exists a path~$X\in \Phi$ 
 such that $\sum_{e\in X} C_e=0$. From the structure of the first stage costs, it  follows that it 
  consists of the arcs: $(s_1,x_1^{\text{in}})$, $(x_n^{\text{out}}, t_1)$, $(t_1,s_2)$ and $(s_2,t_2)$
  and the subpath traversing arcs of the $n$ $\Xset$ variables subgraphs.
 This subpath uniquely defines 
 the 0-1 assignment~$\pmb{x}$ of the variables in~$\Xset$, i.e.
 $x_i:=1$ if the subpath contains negative literal subpath of~$x_i$ variable gadget, $x_i:=0$ otherwise.
 Consider any $\Yset'\subset \Yset$ such that $|\Yset'|\leq \Gamma$.
 Let $A'_{\pmb{y}}=\{e_{y_i}\, :\, y_i\in \Yset'\}\subseteq A_{\pmb{y}}$, $|A_{\pmb{y}}|\leq  \Gamma= \Gamma^d$.
 Define  scenario $\pmb{c}\in \widehat{\mathcal{U}}(\Gamma^d)$ as follows:
 $c_e:=\overline{c}_e$ if $e\in A'_{\pmb{y}}$, $c_e:=\hat{c}_e$ if $e\in A\setminus  A'_{\pmb{y}}$.
 Since $(G,\pmb{C},\widehat{\mathcal{U}}(\Gamma^d),k=3,\gamma=0)$,
 there exists a simple path $Y\in \Phi^{\mathrm{excl}}(X,k)$ such that $\sum_{e\in Y} c_e=0$.
 From the structure of~$\pmb{c}$ and the fact that $Y\in \Phi^{\mathrm{excl}}(X,k)$, we have that
 none of the arcs $(s_1,x_1^{\text{in}})$, $(x_n^{\text{out}}, t_1)$, $(s_2,t_2)$ and arcs from~$A'_{\pmb{y}}$ belong to~$Y$.
 It has to traverse $(t_1,s_2)$ and contain, as subpaths,
 two simple node-disjoint paths~$\pi_1$ going from~$s_1$ to~$t_1$ through all switch, variable and clause gadgets,
 and~$\pi_2$, going from~$s_2$ to~$t_2$.
 Thanks to the blocking property of switches, the path~$\pi_1$ uniquely defines the 0-1 assignment $(\pmb{y},\pmb{z})$,
 namely, $y_i:=1$ ($z_i:=1$) if $e_{\overline{y}_i}\ \in \pi_1$ ($e_{\overline{z}_i} \in \pi_1$); 
 $y_i:=0$ ($z_i:=0$) if $e_{y_i}\in \pi_1$ ($e_{z_i}\in \pi_1$), $i\in [n]$.
 Obviously, $\pmb{y}|_{\Yset'}=\pmb{0}$.
 Since all clause gadgets are traversed according to the uniquely defined assignment $(\pmb{x}, \pmb{y}, \pmb{z})$,
 this assignment satisfies~$\digamma(\pmb{x},\pmb{y},\pmb{z})$
 and $(\Xset,\Yset,\Zset,\digamma,\Gamma)$ is a Yes-instance.
\end{proof}

By Observation~\ref{obext} and Theorem~\ref{drrexcl}, we get the following corollary:
\begin{cor}
\label{cdrrexcl}
\textsc{Rec Rob SP} under $\mathcal{U}(\Gamma^d)$
 with the neighborhood $\Phi^{\mathrm{excl}}(X,k)$
 is $\Sigma^p_3$-hard and $\Sigma^p_3$-hard to approximate even if $k=3$.
\end{cor}

\begin{thm}\label{drrsym}
    \textsc{Decision-Rec Rob SP}$(\Gamma^d)$ with the neighborhood $\Phi(X,k):=\Phi^{\mathrm{sym}}(X,k)$ is $\Sigma^p_3$-complete.
\end{thm}
\begin{proof}
By Lemma~\ref{lrrpi2p}, \textsc{Decision-Rec Rob SP}$(\Gamma^d)$ with $\Phi^{\mathrm{sym}}(X,k)$ is in~$\Sigma^p_3$.
We prove $\Sigma^p_3$-hardness by modifying the reduction from the proof of Theorem~\ref{drrexcl}
to adapt it from the exclusion to the symmetric-difference neighborhood.

In the recoverable robust problem the first-stage path~$X$ is not fixed,
and we must ensure that the second-stage path~$Y$ traverses the $\Xset$-variable gadgets using exactly the same arcs as~$X$.
In the exclusion neighborhood we handle this by allowing $Y$ to exclude only~$3$ arcs of~$X$;
the cost structure and the cost threshold then force those arcs to lie outside the $\Xset$-variable gadgets.
The symmetric-difference neighborhood, however, counts both excluded and included arcs,
so we must additionally prevent~$Y$ from exchanging any arcs with~$X$ inside the $\Xset$-variable gadgets.

To force the second-stage path~$Y$ to traverse the $\Xset$-variable gadgets using exactly the same arcs as~$X$,
we modify the first step of the reduction, in which the graph~$G = (V, A)$ is constructed,
so that all feasible first-stage paths have the same number of arcs and, independently, all feasible second-stage paths have the same number of arcs.
To this end we modify the variable gadget.
Recall that the switch gadgets replace the arcs corresponding to the literals of the variables occurring in the clauses with paths of~$5$ arcs,
e.g., if a variable~$x_1$ occurs in the clause~$\mathcal{C}_1$ as its positive literal,
then the~$(x_1^0, x_1^1)$ arc is connected to one of the~$(\mathcal{C}_1^{\text{in}}, \mathcal{C}_1^{\text{out}})$ arcs in the clause-gadget for~$\mathcal{C}_1$.
Due to the switch gadgets, the number of arcs in a variable gadget depends on the formula~$\digamma$, and,
moreover, a path may use a different number of arcs depending on whether it traverses the gadget along its upper or its lower subpath.
To make the gadgets uniform, after the switch gadgets have been connected to the variable and clause gadgets,
we replace each arc inside the upper and lower subpaths of a variable gadget that the switch gadgets left unchanged,
except the in- and out-arc of each subpath (its first and last arc), with a path of~$5$ arcs.
While the two in/out arcs are kept intact, every other arc of these subpaths (arcs encoding an assignment of a given variable) is then expanded into a path of~$5$ arcs,
so each assignment-encoding subpath has exactly~$5m+2$ arcs, where~$m$ is the number of clauses in~$\digamma$,
regardless of whether the variable is set to true or false.
We assign the first- and second-stage costs of the newly inserted arcs so that each replacement
path has the same total cost as the arc it replaces; hence the cost of every path is unchanged and the cost threshold~$0$ is preserved as well.
It follows that a zero-cost first-stage path encoding a valid assignment to the variables in~$\Xset$,
at this step of the construction, has $|X| = n(5m+2) + 4 = 5mn + 2n + 4$ arcs, where~$n = |\Xset|$.
Correspondingly, a second-stage path~$Y$ encoding a proper assignment to all variables has,
at this step of the construction, $|Y| = 15nm + 6n + 35m + 3$ arcs.
Such a second-stage path excludes exactly 3 arcs of~$X$,
namely, the~$(s_1, x_1^\text{in})$, $(x_{n}^{\text{out}}, t_1)$, and~$(s_2, t_2)$ arcs.
The number of included arcs, assuming $Y$ traverses exactly the same subpaths in $\Xset$-variable gadgets as $X$,
equals $10nm + 4n + 35m + 2$.

The second modification consists in adding a subpath~$t_1$-$s_2$
of~$|A| - 1$ arcs in place of the~$(t_1, s_2)$ arc in the second step of the reduction.
The number of arcs, due to the added~$\Xset$-variable gadgets and expanding the arcs in them, is now:
$|A| = 30nm + 12n + 57m + 6$; we use $|A| - 1$ arcs, i.e., all arcs except the replaced~$(t_1, s_2)$ arc.
Just as in the exclusion neighborhood case,
replacing this arc with that subpath allows us to ensure that both paths~$X$ and~$Y$ must contain the~$t_1$-$s_2$ subpath.
Now, by setting the neighborhood-size parameter carefully we ensure that the second-stage path properly traverses the graph and encodes a Boolean assignment.
We set~$k = 10nm + 4n + 35m + 5$, i.e., the neighborhood-size parameter is exactly the sum of the arcs included and excluded
by a second-stage path that traverses the same $\Xset$-subpaths as~$X$.
This value is tight: since every assignment-encoding subpath has the same length~$5m+2$,
deviating from~$X$ on an $\Xset$-variable gadget excludes one $5m+2$-arc subpath and includes another,
raising the symmetric difference of~$X$ and~$Y$ by~$2(5m+2)$ and so exceeding~$k$,
while any choice of the~$\Yset$- and~$\Zset$-assignment by the second-stage path contributes the same number of included arcs,
i.e., yields the same symmetric difference~$k$.
Hence every second-stage path within the neighborhood traverses exactly the same $\Xset$-subpaths as~$X$,
so the first-stage assignment~$\pmb{x}$ is inherited by the second-stage path,
and the rest of the proof follows exactly as in the~$\Phi^{\mathrm{excl}}(X,k)$ neighborhood case.
\end{proof}

Observation~\ref{obext} and Theorem~\ref{drrsym} lead to the following result:
\begin{cor}
\label{cdrrsym} 
\textsc{Rec Rob SP} under $\mathcal{U}(\Gamma^d)$
 with the neighborhood $\Phi^{\mathrm{sym}}(X,k)$
 is $\Sigma^p_3$-hard and $\Sigma^p_3$-hard to approximate.
\end{cor}

\section{Conclusions}

In this paper the computational complexity of the recoverable robust shortest path problem with discrete recourse was investigated.
For the arc exclusion and arc symmetric difference neighborhoods the problem turned out to be $\Sigma_3^p$-hard
and the inner adversarial problem turned out to be $\Pi_2^p$-hard.
There are still some open questions regarding the problems under consideration. 
Namely, the \textsc{Adv SP} and \textsc{Rec Rob SP} problems with the arc inclusion neighborhood in
general digraphs are known to be strongly NP-hard~\cite{BKS95} and strongly NP-hard as well as 
not approximable~\cite{B12,NO13}, respectively.
However, it remains unknown whether these hardness results can be further strengthened. In this context,
\textsc{Decision-Adv SP} belongs to coNP, while \textsc{Decision-Rec Rob SP} is in $\Sigma_2^p$.

\subsubsection*{Acknowledgements}
The authors were supported by
 the National Science Centre, Poland, grant 2022/45/B/HS4/00355.

%\bibliographystyle{abbrv}
%\bibliography{robust}

\appendix

\section{Appendix}
\label{dod}

We will use the following problem:

\begin{description}
\item{$\forall \exists$\textsc{3CNF-SAT}}
\item{Input:} Two disjoint sets $\Pset$ and $\Qset$ of Boolean variables that take the values 0 and 1,
a Boolean formula $\Psi(\Pset,\Qset)$ in conjunctive normal form (3CNF), where every clause consists of three literals.
\item{Question:} Is it true that for all 0-1 assignments~$\pmb{p}$ of variables in~$\Pset$ 
there exists a 0-1~assignment~$\pmb{q}$ of variables in~$\Qset$ such
that $\Psi(\pmb{p},\pmb{q})$ is satisfied?
\end{description}
$\forall \exists$\textsc{3CNF-SAT} is $\Pi^p_2$-complete, even for $|\Pset|=|\Qset|$, (see, e.g.,~\cite{SU02}).

\begin{proof}[The proof of Theorem~\ref{thm01}]
Obviously, $\forall (\Gamma) \exists$\textsc{CNF-SAT} is in $\Pi^p_2$,
since it can be expressed in the form
$\forall (\Yset' \subseteq \Yset\,:\, |\Yset'|\leq \Gamma) 
(\exists (\pmb{y},\pmb{z}) \in \{0,1\}^{|\Yset|\times |\Zset|}\,:\, 
\pmb{y}|_{\Yset'}=\pmb{0})\digamma(\pmb{y},\pmb{z})$. 
We need to show that it is $\Pi^p_2$-hard by providing a reduction from $\forall \exists$\textsc{3CNF-SAT}
to $\forall (\Gamma) \exists$\textsc{CNF-SAT}, which 
 is adapted from~\cite[Theorem~1]{GLW24} to our case (some parts are nearly verbatim).
 
Given an instance $(\Pset,\Qset,\Psi)$  of  $\forall \exists$\textsc{3CNF-SAT}, namely,
 two  disjoint sets of Boolean variables, $\Pset=\{p_1,\ldots,p_n\}$ and $\Qset=\{q_1,\ldots,q_n\}$,
 a Boolean formula  $\Psi$ in CNF, $\Psi=\mathcal{C}_1\wedge\cdots\wedge\mathcal{C}_m$, where
 each clause $\mathcal{C}_i$, $i\in[m]$, contains three literals, we build an instance $(\Yset,\Zset,\digamma,\Gamma)$ of $\forall (\Gamma) \exists$\textsc{CNF-SAT} in two steps.
 First, we define two disjoint sets of Boolean variables:  
 $\Yset=\{y^{\mathrm{P}}_1,\ldots,y^{\mathrm{P}}_n, y^{\mathrm{N}}_1,\ldots,y^{\mathrm{N}}_n\}$
 and $\Zset=\{z_1,\ldots, z_n\}\cup\{s_1,\ldots,s_n\} \cup \{s\}$
 and replace in the formula~$\Psi$ each positive literal~$p_j$ by $y^{\mathrm{P}}_j$,
 each negative literal~$\overline{p}_j$ by~$y^{\mathrm{N}}_j$,
 each positive literal~$q_j$ by $z_j$ and each negative literal~$\overline{q}_j$ by~$\overline{z}_j$, $j\in[n]$.
 The resulting  formula~$\Psi'=\mathcal{C}'_1\wedge\cdots\wedge\mathcal{C}'_m$ is in 3CNF and now
 depends on new variables~$y^{\mathrm{P}}_j$, $y^{\mathrm{N}}_j$ and $z_j$, $j\in[n]$.
 In the second step, we  form the formula~$\digamma$, which is based on the modified~$\Psi'$ in
 which to each clause~$\mathcal{C}'_i$ we add the additional positive literal~$s$, i.e.
 \begin{equation}
 \mathcal{C}''_i:=(\mathcal{C}'_i\vee s),  \,i\in [m]. \label{4lc}
 \end{equation}
 Now each $ \mathcal{C}''_i$, $i\in [m]$, consists of exactly four literals.
Then, we add to~$\digamma$ $2n$ two literal clauses of the following form:
\begin{equation}
(y^{\mathrm{P}}_j \vee \overline{s}_j) \text{ and } (y^{\mathrm{N}}_j \vee \overline{s}_j), \,j\in [n].\label{2lc}
\end{equation}
Finally, we add to~$\digamma$ the following clause
\begin{equation}
(\overline{s} \vee  s_1  \vee \cdots \vee  s_n) \label{1lc}
\end{equation}
and set $\Gamma=n$.
Thus, the constructed instance $(\Yset,\Zset,\digamma,\Gamma)$ of 
$\forall (\Gamma) \exists$\textsc{CNF-SAT} contains $4n+1$ variables and the formula~$\digamma$
consists of $m+2n+1$ clauses.
To meet the assumption~$|\Yset|=|\Zset|$, it is enough to add to~$\Yset$ a ``dummy" variable that does not appear in any clause of~$\digamma$.
The constructed formula contains clauses of more or fewer literals than 3
    but such formulas can always be normalized to 3CNF using the same techniques as used when reducing $\textsc{Satisfiability}$ to $\textsc{3-Sat}$ \cite{KR74}.

It remains to show that 
an instance $(\Pset,\Qset,\Psi)$  of  $\forall \exists$\textsc{3CNF-SAT}
 is a Yes-instance if and only if  the built instance $(\Yset,\Zset,\digamma,\Gamma)$  of $\forall (\Gamma) \exists$\textsc{CNF-SAT}
 is a Yes-instance. 
 
($\Rightarrow$)
 Assume that instance $(\Pset,\Qset,\Psi)$ is a Yes-instance.
 Let $\Yset'\subseteq \Yset$ be an arbitrary but fixed subset such that $|\Yset'|\leq \Gamma=n$.
 We need to consider two cases.
 
 Case 1: $|\Yset'|=n$ and for every $j\in [n]$ exactly one of the variables $y^{\mathrm{P}}_j$ and $y^{\mathrm{N}}_j$
 belongs to~$\Yset'$. Thus, in this case, $\Yset'$ uniquely defines the 0-1 assignment~$\pmb{p}$ of  variables in~$\Pset$, i.e.
 $p_j:=0$ if $y^{\mathrm{P}}_j\in \Yset'$ and $p_j:=1$ if $y^{\mathrm{N}}_j\in \Yset'$, $j\in[n]$, and vice versa.
 Since $(\Pset,\Qset,\Psi)$ is a Yes-instance, there exists a 0-1 assignment $\pmb{q}$  of  variables in~$\Qset$
 such that $\Psi(\pmb{p},\pmb{q})$ is satisfied. Therefore,
 there exist the 0-1 assignment~$\pmb{y}$  of  variables in~$\Yset$, where
 $y^{\mathrm{P}}_j:=1$ and $y^{\mathrm{N}}_j:=0$ if $p_j=1$; 
 $y^{\mathrm{P}}_j:=0$ and $y^{\mathrm{N}}_j:=1$ if $p_j=0$, $j\in[n]$,
 of course $\pmb{y}|_{\Yset'}=\pmb{0}$, and the assignment~$\pmb{z}\in\{0,1\}^n$,
 where $z_j=1 \leftrightarrow q_j=1$, $j\in[n]$ such that $\Psi'(\pmb{y},\pmb{z})$ is satisfied
 and hence every $\mathcal{C}'_i$, $i\in[m]$, is satisfied. In consequence,
 all the clauses~(\ref{4lc}) are satisfied.
 In order to satisfy all the clauses~(\ref{2lc}) and the clause~(\ref{1lc}), we set $s_j=0$ for every~$j\in[n]$
 and $s=0$.
 Thus,  $\digamma(\pmb{y},\pmb{z})$ is satisfied.
 
 Case 2: there exists $j^{*}\in [n]$ such that neither of the variables $y^{\mathrm{P}}_{j^{*}}$ and $y^{\mathrm{N}}_{j^{*}}$ belongs to~$\Yset'$.
 In this case, $\Yset'$ does not define a 0-1 assignment~$\pmb{p}$ of variables in~$\Pset$.
 We now define an 0-1 assignment $(\pmb{y},\pmb{z})$ satisfying~$\digamma$, namely, $\pmb{y}|_{\Yset'}:=\pmb{0}$,
 $y^{\mathrm{P}}_{j^{*}}:= 1$, $y^{\mathrm{N}}_{j^{*}}:=1$, $s_{j^{*}}:=1$, $s_j:=0$ for every $j^{*}\not=j$ and $s:=1$.
 The rest of the components in  $(\pmb{y},\pmb{z})$ are arbitrary.
 It is easily seen that the  clauses~(\ref{4lc}),~(\ref{2lc}) and~(\ref{1lc}) are satisfied
 for the assignment~$(\pmb{y},\pmb{z})$ and  $\digamma(\pmb{y},\pmb{z})$ is satisfied.
 
 ($\Leftarrow$)
 Assume that instance $(\Yset,\Zset,\digamma, \Gamma)$ is a Yes-instance.
 Let $\pmb{p}$ be an arbitrary 0-1 assignment of variables in $\Pset$.
 This assignment uniquely defines the subset $\Yset'\subseteq \Yset$ such that $|\Yset'|=n$, i.e.
 $\Yset'=\{y^{\mathrm{P}}_j\,:\, p_j=0, j\in [n]\} \cup \{y^{\mathrm{N}}_j\,:\, p_j=1, j\in [n]\}$.
  Since $(\Yset,\Zset,\digamma, \Gamma)$ is a Yes-instance, there exists a 0-1 assignment $(\pmb{y},\pmb{z})$  of  variables in~$\Yset$ 
  and~$\Zset$
 such that $\pmb{y}|_{\Yset'}=\pmb{0}$ satisfying~$\digamma$.
 Since the clauses~(\ref{2lc}) and~(\ref{1lc}) are satisfied,
 it is easy to check that $s_j=0$ for every $j\in [n]$  and $s=0$.
 From the assignment~$s=0$ and the fact that all the clauses~(\ref{4lc}) are satisfied,
 it follows that for each $i\in [m]$ the clause $\mathcal{C}'_i$ has to be satisfied by $(\pmb{y},\pmb{z})$.
 Consider every clause $\mathcal{C}'_i$, $i\in [m]$.
 If $\mathcal{C}'_i$ is satisfied by the literal $y^{\mathrm{P}}_j$ ($y^{\mathrm{P}}_j=1$),
 then, by the definition of~$\Yset'$, $y^{\mathrm{P}}_j\not \in \Yset'$  and the corresponding $p_j=1$ and the clause 
 $\mathcal{C}_i$ is satisfied by~$p_j$. Likewise, 
 if $\mathcal{C}'_i$ is satisfied by a literal $y^{\mathrm{N}}_j$ ($y^{\mathrm{N}}_j=1$),
 then again, by the definition of~$\Yset'$, $y^{\mathrm{N}}_j\not \in \Yset'$ and the corresponding $p_j=0$, therefore, the clause 
 $\mathcal{C}_i$ is satisfied by~$\overline{p}_j$.
 Clearly, a 0-1 assignment~$\pmb{q}$ can be defined as follows $z_j=1 \leftrightarrow q_j=1$, $j\in[n]$.
 Hence, $\Psi(\pmb{p},\pmb{q})$ is satisfied.
 \end{proof}

\end{document}